%% file: ecrc-template.tex
\documentclass[3p,times]{elsarticle}

\usepackage{ecrc}

\volume{00}

\firstpage{1}

\journalname{Procedia Computer Science}

\runauth{}

\jid{procs}

\jnltitlelogo{Procedia Computer Science}

\usepackage{amssymb}

\usepackage[figuresright]{rotating}

\input{commands}
\begin{document}

\begin{frontmatter}



\dochead{}

\title{Improved Algorithm for Permutation Testing}


\author{
  Xiaojin Zhang\textsuperscript{\rm 1}\\
  \textsuperscript{\rm 1}Huazhong University of Science and Technology\\
  xiaojinzhang@hust.edu.cn
  }

\address{}

\begin{abstract}
   For a permutation $\pi: [K]\rightarrow [K]$, a sequence $f: \{1,2,\cdots, n\}\rightarrow \mathbb R$ contains a $\pi$-pattern of size $K$, if there is a sequence of indices $(i_1, i_2, \cdots, i_K)$ ($i_1<i_2<\cdots<i_K$), satisfying that $f(i_a)<f(i_b)$ if $\pi(a)<\pi(b)$, for $a,b\in [K]$. Otherwise, $f$ is referred to as $\pi$-free. For the special case where $\pi = (1,2,\cdots, K)$, it is referred to as the monotone pattern. \cite{newman2017testing} initiated the study of testing $\pi$-freeness with one-sided error. They focused on two specific problems, testing the monotone permutations and the $(1,3,2)$ permutation. For the problem of testing monotone permutation $(1,2,\cdots,K)$, \cite{ben2019finding} improved the $(\log n)^{O(K^2)}$ non-adaptive query complexity of \cite{newman2017testing} to $O((\log n)^{\lfloor \log_{2} K\rfloor})$. Further, \cite{ben2019optimal} proposed an adaptive algorithm with $O(\log n)$ query complexity. However, no progress has yet been made on the problem of testing $(1,3,2)$-freeness. In this work, we present an adaptive algorithm for testing $(1,3,2)$-freeness. The query complexity of our algorithm is $O(\epsilon^{-2}\log^4 n)$, which significantly improves over the $O(\epsilon^{-7}\log^{26}n)$-query adaptive algorithm of \cite{newman2017testing}. This improvement is mainly achieved by the proposal of a new structure embedded in the patterns.
\end{abstract}


\end{frontmatter}


\section{Introduction}
We consider the problem of testing forbidden patterns. We say that $f: [n]\rightarrow \mathbb R$ contains a $\pi$-pattern of size $K$, if there is a sequence of indices $(i_1, i_2, \cdots, i_K)$ ($i_1<i_2<\cdots<i_K$), satisfying that $f(i_a)<f(i_b)$ if $\pi(a)<\pi(b)$, for $a,b\in [K]$. Otherwise, $f$ is referred to as $\pi$-free. For the special case where $\pi = (1,2,\cdots, K)$, it is referred to as the monotone pattern. Given a pattern $\pi$, we are interested in designing a tester with one-sided error. That is, the tester accepts if the input is $\pi$-free, and rejects with probability at least $2/3$ if the input differs in at least $\epsilon n$ values from every $\pi$-free sequence. \cite{ergun2000spot} proposed a non-adaptive algorithm for testing $(1,2)$-freeness with one-sided error, the query complexity of which is $O(\log n)$. They also presented a lower bound of $\Omega(\log n)$ for non-adaptive algorithms. A lower bound of $\Omega(\log n)$ for adaptive algorithm was further shown by \cite{fischer2004strength}. \cite{newman2017testing} considered monotone pattern testing for any $K$. They proposed a non-adaptive algorithm for testing $(1,2, \cdots, K)$-pattern, the query complexity of which is $(\log n)^{O(K^2)}$. 
Subsequent works seek to propose more efficient algorithms taking advantage of some smart combinatorial properties. \cite{ben2019finding} provided a non-adaptive algorithm with query complexity $O((\log n)^{\lfloor\log K\rfloor})$, which matched the lower bound. Later, \cite{ben2019optimal} presented an adaptive algorithm for testing monotone pattern, the query complexity of which is $O(\log n)$ and is optimal. The problem of testing monotone patterns is now well-understood. However, it remains open whether the non-monotone patterns could be tested in a more efficient way.

Viewing from the wide applications of forbidden patterns in combinatorics and in time series analysis, \cite{newman2017testing} proposed an adaptive tester with one-sided error for a typical forbidden pattern $(1,3,2)$. This kind of non-monotone forbidden patterns characterizes the permutations that could be generated using the identity permutation with a Gilbreath shuffle. It is worth noting that their tester finds the third element by considering the following two cases: whether the distance between $f$ and the monotone function class is larger than $\log^{-5} n$ or smaller than $\log^{-5} n$. Designing a tester according to these two cases could simplify the analysis. However, it might make the resulting tester inefficient. One possible approach towards designing a more efficient tester for $(1,3,2)$-pattern is finding some structures embedded in these patterns. Luckily, we observe that this way towards improvement is feasible by designing a new greedy approach for generating disjoint $(1,3,2)$-patterns. Importantly, the patterns generated using this approach has a monotone structure, which facilitates the design of a more efficient approach for finding the $(1,3,2)$-pattern. Besides, \cite{ben2019optimal} investigated the problem of testing monotone patterns. Their approach achieves optimal sample complexity. However, it is not clear how to directly apply their approach for designing a tester for $(1,3,2)$-pattern with improved query complexity. With this newly found monotone structure, we could borrow the idea from \cite{ben2019optimal}, and design efficient approaches for identifying $(1,3,2)$-pattern from a dual point of view.

We provide an adaptive algorithm for testing $(1,3,2)$-freeness with improved query complexity. By observing that a monotone structure could be constructed using the third element of $(1,3,2)$-pattern, we design an algorithm that achieves $\tilde O(\epsilon^{-2}\log^4 n)$ query complexity. The advantage of this structure is that binary search could be directly performed in an efficient way. Another key step towards designing an efficient binary search approach is to identify the element that belongs to the monotone sequence. The naive approach of performing random sampling could only guarantee a small success probability, thereby leading to a rather large query complexity. We overcome this challenge by observing that the newly found monotone structure implies that the algorithm for testing monotone pattern with length $2$ could be used to identify the element in an efficient way. Thereby achieving an improved query complexity.

\subsection{Main Results}

\begin{theorem}\label{thm: test_132_main_result}
For any $\epsilon>0$, there exists an adaptive algorithm that, given query access to a function $f:[n]\rightarrow \mathbb R$ which is $\epsilon$-far from $(1,3,2)$-free, outputs a $(1,3,2)$ subsequence of $f$ with probability at least $9/10$. The query complexity of this algorithm is $O(\epsilon^{-2}\log^4 n)$. 
\end{theorem}
The query complexity of the adaptive algorithm proposed by \cite{newman2017testing} is $O(\epsilon^{-7}\log^{26}n)$. Our algorithm thus significantly improves over the previous work.

\begin{table}[!htp]
  \centering
  \caption{A summary of works on permutation testing}
  \label{tab: three environments}
    \begin{tabular}{cccccc}
    \toprule
    \hline
    Algorithm & Permutation & Query complexity & Type of algorithm & Type of error\cr
    \midrule
  \cite{ergun2000spot} & $(1,2)$ & $O(\log n/\epsilon)$ & non-adaptive & one-sided\cr
  \cite{fischer2004strength} & $(1,2)$ & $\Omega(\log n)$ & adaptive & one-sided\cr
  \cite{newman2017testing} & $(1,3,2)$ & $O(\epsilon^{-7}\log^{26}n)$ & adaptive & one-sided\cr
  \cite{newman2017testing} & $(1,3,2)$ & $\Omega(\sqrt{n})$ & non-adaptive & one-sided\cr
  \cite{newman2017testing} & $(1,2,\cdots, K)$ & $(\epsilon^{-1}\log n)^{O(K^2)}$ & non-adaptive & one-sided\cr
  \cite{ben2018improved} & general $\pi$ of length $K\ge 3$ & $\Theta(\epsilon^{-1/K} n^{1-1/K})$ & non-adaptive & one-sided\cr
  \cite{ben2019finding} & $(1,2,\cdots, K)$ & $\Theta((\log n)^{\lfloor \log_{2} K\rfloor})$ & non-adaptive & one-sided\cr
  \cite{ben2019optimal} & $(1,2,\cdots, K)$ & $\Theta(\log n)$ & adaptive & one-sided\cr
  Our work & $(1,3,2)$ & $O(\epsilon^{-2}\log^4 n)$ & adaptive & one-sided\cr
     \hline
    \bottomrule
    \end{tabular}
\end{table}

\subsection{Preliminaries}

A sequence $f$: $\{1,2,\cdots, n\}\rightarrow \mathbb{R}$ is \textit{$\pi$-free}, if there are no subsequences of $f$ with order pattern $\pi$. We say that a sequence $f$ is \textit{$\epsilon$-far from $\pi$-free} if any $\pi$-free function $g: [n]\rightarrow \mathbb{R}$ satisfies $P_{i\in [n]} [f(i)\neq g(i)]\ge\epsilon$. Let $k, \epsilon$ be two fixed parameters. Given query access to a function $f: [n]\rightarrow \mathbb{R}$, the tester with \textit{one-side error} satisfies that: If $f$ is $\epsilon$-far from $\pi$-free, then the algorithm could output a $\pi$ pattern with probability at least $9/10$. If $f$ is $\pi$-free, then the algorithm returns accept. For any two tuples $(i_1, i_2, \dots, i_K)$ and $(j_1, j_2, \dots, j_K)$, they are referred to as \textit{disjoint} if and only if $i_h\neq j_m$, for any $h, m \in\{1,2,\dots, K\}$. Let $(i_1, i_2, \cdots, i_{K})$ be a monotone subsequence of $f$. If $c$ is the smallest integer such that $i_{c+1} - i_c\ge i_{b+1} - i_b$ for all $b\in [K-1]$, then we say that $(i_1, i_2, \cdots, i_{K})$ is an \textit{$c$-gap} subsequence, and $c$ is referred to as the gap of this tuple.

\subsection{Technical Overview}

We propose a simple algorithm for testing $(1,2)$-freeness. This tester is both simple and efficient. Besides, it could be further used to improve the query complexity of the tester for $(1,3,2)$-tuple. The proposal of this simple algorithm is inspired by the structure found by \cite{ben2019finding}, for $f: [n]\rightarrow\mathbb{R}$ that is $\epsilon$-far from monotone pattern $(1,2,\cdots, K)$. Let $\delta_{l,i}$ represent the fraction (also referred to as \textit{density}) of $c$-gap tuples that belongs to the range $[l-K\cdot 2^{i+1}, l+K\cdot 2^{i+1}]$ (referred to as a \textit{block}).
From the dual point of view, it holds that the expected cumulative density $\mathbb E_{l\sim [n]}[\sum_{i\in [\eta]}\delta_{l,i}]$ is at least $\epsilon$, where $\eta = O(\log n)$. With this equation and setting $K = 2$, we analyze according to two cases based on the summands $\mathbb E_{l\sim [n]}[\sum_{i\in [\eta]}\delta_{l,i}]$.

In the first case, the summands are spread over distinct blocks. That is, $\delta_{l,i}\le\epsilon$ for all $i$. The density of each block is rather small, while the summation of the density of all the blocks are rather large. This property ensures that at least two indices could be found from distinct blocks separately, thereby forming a $(1,2)$-tuple. In the second case, the summands are concentrated on few blocks. That is, there exists an integer $i$ such that $\delta_{l,i}\ge\epsilon$. Combined with the symmetric property, the indices found from the left block ($[l-K\cdot 2^{i+1}, l]$) and the right block ($[l, l+K\cdot 2^{i+1}]$) could be concatenated to form a $(1,2)$-pattern with high probability.

Now we will introduce the tester for $(1,3,2)$-pattern. The first key ingredient of our algorithm is that we utilize the monotone structure constructed by the third element of $(1,3,2)$-permutation. Specifically, if the algorithm constructs the subsequence by greedily selecting a maximum index, satisfying that this index is a part of the $(1,3,2)$-pattern. Then it turns out that the set of the third elements at the right side of a fixed index $l$ forms a monotone sequence (Lemma \ref{lem:ijk_monotone}). With this monotone structure, our algorithm could then conduct random binary search efficiently. Another challenge is to identify the element that belongs to the monotone sequence. We overcome this challenge by using the tester for two-element pattern. Naively, random sampling is used in order to find the monotone element. However, this approach might lead to large query complexity due to its small success probability. With the property of the monotone structure, we could design efficient algorithm to identify the element that belongs to the monotone structure based on the tester for length $2$ pattern.


\section{Review of previous work}

The problem of testing monotonicity of a function over a partially ordered set has been widely investigated. Various works consider the case when the partially ordered set is the line $[n]$ (e.g., \cite{fischer2004strength, belovs2018adaptive, pallavoor2017parameterized,ben2018testing}), the Boolean hypercube $\{0,1\}^d$ (e.g., \cite{dodis1999improved,blais2012property,briet2012monotonicity}), and the hypergrid $[n]^d$ (e.g., \cite{blais2014lower, chakrabarty2013optimal,black2018d}). Other related problems include estimating the distance to monotonicity and the length of the longest increasing subsequence (e.g., \cite{parnas2006tolerant, ailon2007estimating, gopalan2007estimating, saks2010estimating}).

Another line of works focus on identifying order patterns in sequences and permutaions (e.g., \cite{fox2013stanley, guillemot2014finding}). \cite{newman2017testing} propose a non-adaptive algorithm for testing monotone patterns, and show that the non-adaptive query complexity is $\Omega(\sqrt{n})$ for any non-monotone pattern. The query complexity of the non-adaptive algorithm proposed by \cite{newman2017testing} for testing $(1,2, \cdots, K)$-pattern is $(k\epsilon^{-1}\log n)^{O(K^2)}$. Later, \cite{ben2019optimal} propose an adaptive algorithm for testing monotone patterns, the query complexity of which is $\left(K^K \cdot (\log(1/\epsilon))^K \frac{1}{\epsilon}\cdot \log(1/\delta)\right)^{O(K)}\cdot \log n$.

\subsection{Growing Suffixes and Splittable Intervals}
Firstly we introduce some basic definitions for illustrating the key structure used in  \cite{ben2019optimal}.
\begin{definition}[Growing Suffix]
The collection of intervals $S(a) = \{S_t(a)| t\in [h]\}$ is referred to as an $(\alpha, \beta)$-growing suffix starting from $a$, if there exists a subset $D(a) = \{D_t(a)\subset S_t(a)| t\in [h]\}$, such that the following properties hold:
\begin{itemize}
\item  $|D_t(a)|/|S_t(a)|\le\alpha$ for all $t\in [h]$, and $\sum_{t=1}^{h} |D_t(a)|/|S_t(a)|\ge\beta$;
\item $f(b)<f(b')$ if $b\in D_t(a)$ and $b'\in D_{t'}(a)$, where $t, t'\in [h]$ and $t<t'$.
\end{itemize}

\end{definition}

\begin{definition}[Splittable Interval]
Let $I\subset [n]$ be an interval. Let $\alpha, \beta\in (0,1]$ and $c\in[k-1]$. Define $T^{(L)} = \{ (i_1, i_2, \cdots, i_c)\in I^c: (i_1, i_2, \cdots, i_c) \text{ is a prefix of a } k\text{-tuple in } T\},$ and $T^{(R)} = \{ (j_1, j_2, \cdots, j_{k-c})\in I^{k-c}: (j_1, j_2, \cdots, j_{k-c})$ \text{ is a suffix of a } $k\text{-tuple in } T\}$, where $T\subset I^{k}$ is a set of disjoint $(1,2, \cdots, K)$-tuple lying in $I$. $(I, T)$ is $(c, \alpha, \beta)$-splittable if the following properties are satisfied:

\begin{itemize}
\item $|T|/|I|\ge\beta$;

\item $f(i_c)<f(j_1)$ for every $(i_1, i_2, \cdots, i_c)\in T^{(L)}$ and $(j_1, j_2, \cdots, j_{k-c})\in T^{(R)}$;

\item There is a partition of $I$ into three adjacent intervals $L, M, R\subset I$ of size at least $\alpha |I|$, satisfying $T^{(L)}\subset L^c$ and $T^{(R)}\subset R^{k-c}$.
\end{itemize}

A collection of disjoint interval-tuple pairs $(I_1, T_1), (I_2, T_2), \cdots, (I_s, T_s)$ is called a $(c, \alpha, \beta)$-splittable collection of $T$ if each $(I_j, T_j)$ is $(c,\alpha,\beta)$-splittable and the sets $(T_j: j\in [s])$ partition $T$.
\end{definition}

Now we will introduce the robust structure of monotone pattern introduced by \cite{ben2019optimal}.
\begin{lemma}[\cite{ben2019optimal}]\label{lem: structure_lem_robust}
There exist $\alpha\in (0,1)$ and $p>0$ with $\alpha\ge\Omega(\epsilon/K^5)$ and $p\le \text{poly}(K\log(1/\epsilon))$ such that at least one of the following holds. 
\begin{itemize}
  \item There exists a set $H\subset [n]$, of indices that start an $(\alpha, Ck\alpha)$-\textbf{growing suffix}, satisfying $\alpha |H|\ge (\epsilon/p)n$.
  \item There exist an integer $c$ with $1\le c < K$, a set $T$, with $E(T)\subset E(T^0)$, of disjoint length-$k$ monotone subsequences, and a $(c, 1/(6k), \alpha)$-\textbf{splittable collection} of $T$, consisting of disjoint interval-tuple pairs $(I_1, T_1), (I_2, T_2), \cdots, (I_s, T_s)$, such that
\begin{align}
\alpha\sum_{h=1}^s |I_h|\ge (\epsilon/p)|I|.
\end{align}
Moreover, if $J\subset I$ is an interval where $J\subset I_h$ for some $h\in [s]$, $J$ contains at least $(\epsilon/p)|J|$ disjoint $(1,2, \cdots ,K)$-patterns from $T^0$.
\end{itemize}
\end{lemma}

\subsection{The adaptive algorithm for testing $(1,2,\dots, K)$-pattern}

Now let us introduce the adaptive algorithm proposed by \cite{ben2019optimal} for testing $(1,2,\dots, K)$-pattern. They consider two cases, growing suffix and splittable interval. We will firstly introduce the basic definition of these two cases, and the key structure their algorithm depends on. The basic definition of growing suffix and splittable interval, and the key structure are illustrated in Lemma~\ref{lem: structure_lem_robust}.


The core idea of the adaptive tester for monotone-pattern of \cite{ben2019optimal} is as follows: From Lemma~\ref{lem: structure_lem_robust}, $f$ either satisfies the growing suffixes condition, or satisfies the robust splittable intervals condition. If $f$ satisfies the growing suffixes condition, then \cite{ben2019finding} already provide an algorithm with $O(\log n)$ query complexity. Assume $f$ meets the robust splittable intervals condition, the algorithm firstly finds an index $x$ that lies at the left part of a splittable interval, and then uses $y$ to estimate the length of an interval, where $y$ is the maximal index satisfying that $f(y)\ge f(x)$ by sampling from all the blocks with length $2^i$ starting from $x$ for $i\in [\lceil\log n\rceil]$.

If $y$ lies close to the interval, then the algorithm could find two intervals each with density $\epsilon$, and the tuples in these two intervals could be concatenated to form $(1,2, \dots, K)$-pattern. This means that there exists an integer $c$, such that any tuple $(1,2,\dots, c)$ from the left interval could be concatenated with the tuple $(c+1, c+2, \dots, n)$ from the right interval. The ratio of such tuples in both intervals are at least $\epsilon$.


If $y$ lies far from the interval, then the algorithm could not find two intervals each with density at least $\epsilon$ within the splittable interval. Instead, the algorithm constructs $k$ successive intervals with density $\epsilon$, according to the strong structure lemma considering growing suffixes and splittable intervals. By recursively using the monotone pattern tester to find a monotone tuple of length $(K-i)$ from $i$-th interval with values larger than $y$, and a monotone tuple of length $(i+1)$ from $i$-th interval with values smaller than $y$. The algorithm could finally find a $K$-monotone tuple by either concatenating the tuple in the boundary interval with $x$ and $y$ (the tuple in the first interval together with $x$, or the tuple in the last interval together with $y$), or the tuple in $i$-th interval together with that in $(i+1)$-th interval. By induction on $K$, the algorithm could find a $(1, 2, \dots, K)$ pattern with high probability with $O_{K,\epsilon}(\log n)$ query complexity.

\subsection{The adaptive algorithm for testing $(1,3,2)$-freeness}

\cite{newman2017testing} proposed an adaptive algorithm for testing $(1,3,2)$-freeness, the query complexity of this algorithm is shown in the following lemma.

\begin{lemma}[\cite{newman2017testing}]
Given query access to a function $f: [n]\rightarrow R$ that is $\epsilon$-far from $(1,3,2)$-free. There exists an algorithm with arguments $(f, \epsilon)$ that returns a $(1,3,2)$-tuple in $f$ with probability at least $\Omega(\epsilon^3/(\log^{6} n)$. The query complexity is at most $O(\epsilon^{-4}\log^{20}n)$.
\end{lemma}

The algorithm firstly divides $[n]$ into intervals $I_1, I_2, \cdots, I_{\lceil n/W\rceil}$ with equal length $W= 2^w$, where $w$ is selected uniformly at random from $\{0,1,\cdots, \lceil\log n\rceil -1\}$. The main techniques are disclosed in the scenario when $2$-gap tuples dominate and therefore we focus on elaborating on the high-level idea of the algorithm designed for this scenario.
Let $T_{2,w}$ represent the disjoint $2$-gap tuples $(i,j,k)$ satisfying that $|k-j|\ge|j-i|$, and the gap between $j$ and $k$ is $\lfloor\log(k-j)\rfloor = w$. 
If $f: [n]\rightarrow \mathbb R$ is $\epsilon$-far from $\pi$-free, then the density of tuples $(i,j,k)$ in $[n]$ is at least $\Omega(\epsilon)$. Therefore, there exists a $w$, satisfying that the density of $T_{2,w}$ in $[n]$ is at least $\Omega(\epsilon/\log n)$. 

Let $T_{2,w,s}$ represent the tuples $(i,j,k)$ in $T_{2,w}$ with $j$ lies in $I_s$, where $s\in [\lceil n/W\rceil]$. It could be inferred that $|T_{2,w,s}|\ge\Omega(\epsilon W/\log n)$ holds for a large fraction of $s$. For any tuple $(i,j,k) \in T_{2,w,s}$, it could be inferred that $(i,j)\in I_L = I_{s-2}\cup I_{s-1}\cup I_{s}$, and $k\in I_R = I_{s+1}\cup I_{s+2}$. Therefore, the density of $(i,j)$ such that $(i,j,k)\in T_{2,w,s}$ in $I_L$ is at least $\Omega(\epsilon/\log n)$, and the density of $k$ such that $(i,j,k)\in T_{2,w,s}$ in $I_R$ is at least $\Omega(\epsilon/\log n)$. 
Let $I_L' = \{i\in I_L: f(i)> f(i_0)\}$, and $I_R' = \{i\in I_R: f(i) > f(i_0)\}$, where $i_0$ is the index with a smallest $f$-value among the sampled $q$ indices, and $q = 100\epsilon^{-4}\log^{20} n$. Then, the algorithm could firstly find $i, j$ using the monotone testing algorithm, and then find $k$ according to the following two cases:

1. If $f|_{I_R'}$ is $(\log^{-5} n)$-far from monotone. Then the density of disjoint $(2,1)$-pattern in $I_R'$ is at least $(\log^{-5} n)/2$. The algorithm will find $(j,k)$ from $I_R'$ using the non-adaptive algorithm, and then concatenate with $i_0$ to form a $(1,3,2)$-tuple. 

2. If $f|_{I_R'}$ is $(\log^{-5} n)$-close to monotone. The algorithm could firstly find a $(i,j)$ from $I_L'$, satisfying that $(i,j)$ forms a monotone increasing subsequence. Then the algorithm adaptively finds $k$ from $I_R'$, satisfying that $(i,j,k)$ forms a $(1,3,2)$-tuple. Specifically, the algorithm uses noisy binary search to find such a $k$ based on the property that $f|_{I_R'}$ is close to monotone. Note that $I_R'$ is not known to the algorithm. It is guaranteed that $|I_R'|/|I_R|\ge \epsilon/(8\cdot\log n)$, the algorithm could perform binary search over $I_R$ instead.





\section{Improved algorithm for testing $(1,3,2)$-freeness}\label{sec: sec_alg_1_3_2_main}

\subsection{Simple algorithm for testing monotone pattern with length $2$}\label{sec: mono_pattern_test}
Let $U$ be the set of disjoint $(i_1, i_2)$-sequences of $f$. For $a<l<b$, the width of the pair $(a,b)$ is defined as $\lceil\log(b-a)\rceil$. Let $A_{l, t}$ denote the $(i_1, i_2)$-tuples with width $t$, and $l$ lies in the middle third of the interval of $[i_1,i_2]$. Specifically, $A_{l, t} = \{(i_1, i_2)\in U: \text{width}(i_1, i_{2}) = t, l\in [i_1 + (i_{2}-i_1)/3, i_{2} - (i_{2}-i_{1})/3]\}$. For $l\in [n]$, we define the $\CumDensity$ of $l$ as $v_l = \sum_{t=1}^{\lceil\log n\rceil} \delta_{l, t}$, where $\delta_{l, t} = |A_{l, t}|/2^{t+1}$. The following lemma shows that most of the indices have large cumulative density.

\begin{lemma}\label{density_l_epsilon}
If $f$ is $\epsilon$-far from $(1,2)$-free, then the ratio of $l\in [n]$ satisfying that $v_l\ge\epsilon/24$ is at least $\epsilon/24$.
\end{lemma}

Suppose $l\in [n]$ is a an index with high $\CumDensity$. Given this index, we design an algorithm for testing $(1,2)$-freeness with $O(\epsilon^{-1}\log n)$ query complexity. The introduction of this tester further facilitates the design of tester for $(1,3,2)$-pattern. Note that testing $(1,2)$-freeness could be directly reduced to testing $(2,1)$-pattern in the reversed sequence. We now focus on presenting the algorithm for testing $(1,2)$-freeness, the corresponding algorithm is referred to as $\text{TestMonotoneEpoch}_{(1,2)}$. If the algorithm aims at identifying $(2,1)$-tuple, then the checking condition should be modified as identifying a $(2,1)$-tuple instead. The corresponding algorithm for testing $(2,1)$-tuple is referred to as $\text{TestMonotoneEpoch}_{(2,1)}$.




The algorithm for testing $(1,2)$-freeness is illustrated in Algorithm \ref{alg: testing $(1,2)$-freeness each epoch}. The main idea of this algorithm is: If the density is separated among distinct blocks, then the algorithm finds $1$ and $2$ from the right blocks; If the density is concentrated on few blocks, then the algorithm combines $1$ obtained from the left block and $2$ from the right block.

\begin{algorithm}[h]
\caption{$\text{TestMonotoneEpoch}_{(1,2)}(f, [n], \epsilon,l)$}
\label{alg: testing $(1,2)$-freeness each epoch}
\begin{algorithmic}[1]
\FOR{$t = 1, 2, \cdots, \lceil\log (n-l)\rceil$}
\STATE{Consider the intervals $L_t(l) = [l-2^t, l]$ and $R_t(l) = [l, l+2^t]$.}
\STATE{Sample from $R_t(l)$ uniformly at random for $\Theta(1/\epsilon)$ times. Sample from $L_t(l)$ uniformly at random for $\Theta(1/\epsilon)$ times.}
\IF{a $(1,2)$-tuple is found}\label{line: check_tuple}
\STATE{return this tuple.} 
\ENDIF
\ENDFOR\\
\STATE{return FAIL.} 
\end{algorithmic}
\end{algorithm}

Now we are ready to show the correctness and query complexity of our algorithm as shown in Theorem~\ref{lem_monotone_two}. The missing proofs in this section are deferred to \ref{proof_sec_6_1} and \ref{proof_sec_6_2}.

\begin{theorem}\label{lem_monotone_two}
Suppose $l\in [n]$ is a an index with $\CumDensity$ at least $\epsilon/24$. For any $\epsilon>0$, given $l$ and query access to a function $f:[n]\rightarrow \mathbb R$ which is $\epsilon$-far from $(1,2)$-free, Algorithm~\ref{alg: testing $(1,2)$-freeness each epoch} outputs a length-$2$ monotone subsequence of $f$ with probability at least $9/10$. The query complexity of this algorithm is $O(\epsilon^{-1}\log n)$. 
\end{theorem}




\subsection{Monotone Structure of $(1,3,2)$-pattern}\label{sec: monotone_1_3_2}
Before introducing the monotone structure for $(1,3,2)$-pattern, we firstly introduce a basic structure that is applicable to general pattern.
\begin{lemma}[Proposition 2.2 of \cite{newman2017testing}]\label{gap_lemma_1}
If $f$ is $\epsilon$-far from $\pi$-free, where $\pi$ is a permutation of $[K]$. Then there exists at least $\epsilon n/K$ disjoint permutation in the form of $\pi$.
\end{lemma}

We generate a set $T$ of disjoint $(1,3,2)$ subsequences of $f$ of length $3$, following the approach GreedyDisjointTuples proposed by \cite{ben2019finding}. The idea of GreedyDisjointTuples is that the algorithm greedily adds a minimum index to the subsequence, if this index could be used to form a $\pi$-pattern. The distinction is that we greedily construct disjoint tuples from right to left, while their approach greedily constructs disjoint tuples from left to right. Since this subroutine is very similar to the subroutine GreedyDisjointTuples, we call it GreedyDisjointTuples$+$. The monotone structure could be extracted from $(1,3,2)$-tuples using this algorithm.

Define $(i, j, k)\prec(i', j', k')$ if (i) $k < k'$ or (ii) $k = k'$ and $j < j'$ or (iii) $k = k', j = j'$ and $i < i'$. Then GreedyDisjointTuples$+$ greedily picks a $(1,3,2)$-tuple from $E(T_0)$ which is maximal under the above partial order in every step. The pseudo-code of this algorithm is shown in Algorithm \ref{alg: GreedyDisjointTuples+}.

\begin{algorithm}[h]
\caption{GreedyDisjointTuples$+$}
\label{alg: GreedyDisjointTuples+}
\begin{algorithmic}[1]
\STATE{Input: a function $f: [n]\rightarrow \mathbb R$, and a set $T_0$ of disjoint $(1,3,2)$ subsequences of $f$}
\STATE{Output: a set $T\subset [n]^{3}$ of disjoint subsequences with monotone structure.}
\STATE{Set $T\leftarrow \emptyset$, let $E(T_0) = \cup_{(i_1,i_2,i_3)\in T_0} \{i_1,i_2,i_3\}$ represent the set of indices of subsequences in $T_0$.}
\WHILE{$E(T_0)\setminus E(T)\neq\emptyset$}
\STATE{Let $k$ be the maximum element in $E(T_0)\setminus E(T)$.}
\STATE{Let $S = \{i_2: (i_1,i_2)\in E(T_0)\setminus E(T), \text{ and } (i_1,i_2,k) \text{ forms a } (1,3,2)\text{-tuple of } f\}$.}
\STATE{// select $j$ as the maximum element in $E(T_0)\setminus E(T)$ for which there exists $i_1$ such that $(i_1, j, k)$ forms a $(1,3,2)$-tuple of $f$.}
\IF{$S\neq\emptyset$}
\STATE{$j \leftarrow \max_{i_2\in S}i_2.$}
\ENDIF
\STATE{Let $M = \{i_1: i_1\in E(T_0)\setminus E(T) \text{ and } (i_1,j,k) \text{ forms a } (1,3,2)\text{-tuple of } f\}$.}
\IF{$M\neq\emptyset$}
\STATE{// select $i$ as the maximum element in $E(T_0)\setminus E(T)$ such that $(i,j,k)$ forms a $(1,3,2)$-tuple of $f$.}
\STATE{$i \leftarrow \max_{i_1\in M}i_1$.}
\ENDIF
\STATE{Update $T\leftarrow T\cup \{(i,j,k)\}$.}
\ENDWHILE
\end{algorithmic}
\end{algorithm}

The following lemma shows the monotone structure obtained using Algorithm~\ref{alg: GreedyDisjointTuples+}, which is the key for proposing an improved tester.

\begin{lemma}\label{lem:ijk_monotone}
Let $T$ be the output of the subroutine GreedyDisjointTuples$+$.
For a given $l$, consider the set of tuples $(i, j, k)\in T$ satisfying that $j\le l \le k$. We denote $A = \{(i, j, k)\in T: j\le l \le k\}$, then the elements $k$ such that $(i, j, k)\in A$ form an increasing sequence.
\end{lemma}

Note that if the algorithm generates disjoint tuples using the approach GreedyDisjointTuples, then the monotone structure illustrated in the above lemma might not hold. 





The property of the tuples generated using GreedyDisjointTuples$+$ shown in the following lemma is similar to that of GreedyDisjointTuples, which holds due to the greedy nature.

\begin{lemma}\label{gap_lemma_2}
Let $f: [n]\rightarrow \mathbb R$, and let $T_0\subset [n]^3$ be a set of disjoint $(1,3,2)$-tuples of $f$ with length $3$. Then there exist $c\in [2]$ and a set $T\subset [n]^3$ of disjoint $c$-gap $(1,3,2)$-tuples with $E(T)\subset E(T_0)$, satisfying that $|T|\ge|T_0|/6$, where $E(T) = \cup_{(i,j,k)\in T} \{i,j,k\}$ is the set of indices of subsequences in $T$.
\end{lemma}

We therefore analyze based on the following two cases: when $1$-gap $(1,3,2)$-tuples dominate and when $2$-gap $(1,3,2)$-tuples dominate.


\subsection{The algorithm for testing $(1,3,2)$-freeness}
In this section, we introduce efficient algorithms with the help of monotone tuple tester illustrated in Section~\ref{sec: mono_pattern_test}.

\subsubsection{$1$-gap tuples dominate}

If $1$-gap tuples dominate, then the algorithm firstly randomly selects an index $l$, and constructs the interval $L_t(l) = [l - 2^t,l]$, $R_t(l) = [l, l + 2^{t}]$ by enumerating $t$ from $1, 2,$ up to $\lceil\log n\rceil$. The algorithm then finds $i$ from $L_t(l)$, and finds $(j,k)$ from $R_t(l)$. The detailed approach is illustrated in Algorithm~\ref{Framework: testing $(1,3,2)$-freeness when 1-gap terminates}.

\begin{algorithm}[h]
\caption{testing $(1,3,2)$-freeness when $1$-gap tuples dominate}
\label{Framework: testing $(1,3,2)$-freeness when 1-gap terminates}
\begin{algorithmic}[1]
\FOR{$m = 1, 2, \cdots, \Theta(1/\epsilon)$}
\STATE{Select $l$ uniformly at random from $[n]$.}
\FOR{$t = 1, 2, \cdots, \lceil\log (n-l)\rceil$}
\STATE{$L_t(l) = [l-2^t, l], R_t(l) = [l, l + 2^{t}]$.}
\IF{\text{TestMonotoneEpoch}$_{(2,1)}(f, R_t(l), \epsilon/\log n, l)$ returns $(j,k)$}
\STATE{Sample $i$ from $L_t(l)$ uniformly at random for $10\lceil\log n\rceil/\epsilon$ number of times.}
\IF{$f(i)<f(k)<f(j)$}
\STATE{return the $(1,3,2)$-tuple $(i,j,k)$.}
\ENDIF
\ENDIF
\ENDFOR
\ENDFOR
\STATE{return FAIL.}
\end{algorithmic}
\end{algorithm}

\subsubsection{$2$-gap tuples dominate}
When $2$-gap tuples dominate, our analysis relies on the clever definition of $\gamma$-deserted element introduced by \cite{newman2017testing}.

\begin{definition}
Given a set $S\subset I$, and a parameter $\gamma\in [0,1]$, an element $i\in S$ is called $\gamma$-deserted, if there exists an interval $J\subset I$ containing $i$ such that $|S\cap J|<\gamma |J|$.
\end{definition}

Firstly we will explain why similar approach proposed for the case when $1$-gap tuples dominate does not work in this case. Suppose we use similar approach as in the first case. Assume that the density of $k$ in $R_t(l)$ is at least $\Omega(\epsilon/\log(n))$, and the density of $(i,j)$ in $L_t(l)$ is at least $\Omega(\epsilon/\log(n))$. The blocks are divided into top, middle and bottom blocks based on the function value. The algorithm could firstly find $k$ such that $k$ lies in the middle block with high probability. Then, it requires to find a $(i,j)$-tuple from $L_t(l)$, satisfying that $f(j)>f(k)>f(i)$. We have no idea about how to design an efficient algorithm to find this form of tuple. For the algorithm $\text{TestMonotone}$ proposed in Section \ref{sec: mono_pattern_test}, the use of growing suffix structure makes it hard to be tuned to satisfy this requirement. Therefore, we design a distinct tester for this case as illustrated in Algorithm~\ref{Framework: testing $(1,3,2)$-freeness}.

For a given $l$, consider the tuples $(j, k)$ that cross $l$. That is, $j\le l \le k$. We have the following key observation: the set of $k$ form a monotone sequence. The monotone structure introduced in Lemma~\ref{lem:ijk_monotone} is the key for the proposal of a more efficient tester for $(1,3,2)$-pattern.


\begin{algorithm}[h]
\caption{testing $(1,3,2)$-freeness when $2$-gap tuples dominate}
\label{Framework: testing $(1,3,2)$-freeness}
\begin{algorithmic}[1]
\FOR{$m = 1,2,\cdots,\Theta(1/\epsilon)$}
\STATE{Select $l$ uniformly at random from $[n]$.}
\FOR{$t = 1, 2, \cdots, \lceil\log (n-l)\rceil$}
\STATE{$L_t(l) = [l-2^{t}, l] $.}
\IF{\text{TestMonotoneEpoch}$_{(1,2)}(f, L_t(l), \epsilon/\log n, l)$ returns $(i,j)$}
\STATE{$R_t(l) = [l, l+2^t]$.}
\IF{BinarySearch ($f, l, L_t(l), R_t(l)$) returns $k$ satisfying that $f(i)<f(k)<f(j)$}
\STATE{return the $(1,3,2)$-tuple $(i,j,k)$.}
\ELSE
\STATE{return FAIL.}
\ENDIF
\ENDIF
\ENDFOR
\ENDFOR
\end{algorithmic}
\end{algorithm}

The basic idea for testing $(1,3,2)$-freeness is: the algorithm firstly finds an index $l$, then finds a block with high density, which is divided into left block and right block by $l$. With these blocks, the algorithm finds $(i,j)$ from left block of $l$, and then adaptively finds $k$ from the right block of $l$.


\begin{itemize}
\item \textbf{Find $l$:} We are interested in the index $l$ satisfying that its $\CumDensity$ is larger than $\epsilon$. The total number of such indexes is $\Omega(\epsilon\cdot n)$.

\item \textbf{Find interval:} Since the $\CumDensity$ of $l$ is at least $\epsilon$. There exists a width $t$ such that the density of $k$ in this block is at least $\delta_{l,t}\ge\epsilon/(\log n) = \epsilon_1$. We will focus on this interval. 



\item \textbf{Find $(i,j)$:} The interval found in the last step ensures that the density of $(i,j)$ in $L_t(l)$ and $R_t(l)$ is at least $\epsilon_1$. By calling the TestMonotone algorithm, $(i,j)$-pattern could be found from $L_t(l)$ using $\tilde O_k(\log n)/\epsilon_1$ number of queries, with high probability.

\item \textbf{Find $k$ adaptively:} Since the density of $k$ in $R_t(l)$ is at least $\epsilon_1$, we could find $k$ adaptively using binary search.
\end{itemize}


Now we will illustrate the detailed binary search approach. Suppose that we find $i$ and $j$ ($f(i)<f(j)$) from $L_t(l)$, and want to complete them to form a $(1,3,2)$-pattern. Then we want to find the corresponding $k$-value ($f(i)<f(k)<f(j)$) using binary search over $R_t(l)$.

Firstly we consider a simpler case as a warm-up. Suppose all the coordinates in $R_t(l)$ form a monotone sequence. Then the standard process of binary search performs as illustrated in Algorithm~\ref{Framework: Standard Binary Search} (deferred to Appendix). At each iteration, suppose the search range is partitioned into three consecutive intervals. The algorithm wants to find an index that belongs to the middle block. Initially, $I = [a,b]$. If $f(x)\le f(i)$, then the block $I$ shrinks to $[x,b]$, the length of this block is less than $2/3 |I|$. If $f(x)\ge f(j)$, then the block $I$ shrinks to $[a,x]$, the length of this block is less than $2/3 |I|$.

However, it could not be guaranteed that all the coordinates in $R_t(l)$ form the monotone sequence. The property that helps is that the density of the monotone sequence in $R_t(l)$ is at least $\epsilon_1$. It guarantees that we could find a coordinate that belongs to the middle block with probability at least $\epsilon_1$, if we randomly sample a coordinate from $R_t(l)$.


Assume that the algorithm randomly sample a coordinate from $R_t(l)$, then with high probability, the coordinate belongs to the middle block. The algorithm then updates the new search range. The length of this interval is less than $2|I|/3$, and the total number of monotone sequence is at least $|T|/3$, where $T$ represents the initial monotone sequence in the search range $R_t(l)$. Then the density of monotone sequence in the newly updated search range $I_1$ becomes $\epsilon_1/2$. In this way, it could only guarantee that the density of monotone sequence in the search range $I_h$ becomes $\epsilon_1/2^h$. The might lead to an algorithm with large query complexity. 

This issue could be resolved by the use of the $\gamma$-deserted elements, which was introduced by \cite{newman2017testing}. From the definition of $\gamma$-deserted elements, the density of the monotone sequence in any interval that contains this element is at least $\gamma$. Therefore, if the third element that we are searching for belongs to the non-$\gamma$-deserted elements, then it could be guaranteed that the density of the monotone sequence is always at least $\gamma$ in all iterations. 


The next question is how do we find the coordinate that belongs to the monotone sequence in the middle block of the search range with high probability? If the algorithm naively performs random sampling to obtain a coordinate from $R_t(l)$, then with probability at least $\gamma$, the coordinate belongs to the monotone sequence in the middle block. However, the lower bound of this probability is very small. How to amplify this probability, and meanwhile ensure the requirement that the coordinate belongs to the monotone sequence in the middle block? We solve this problem by noting that the each coordinate of the monotone sequence has a corresponding coordinate with larger function value in $L_t(l)$. Since the density of monotone sequence in each iteration is at least $\gamma$, $f|_{(I_h\cup L_t(l))}$ is $\gamma$-far from $(2,1)$-free. If $f$ is $\gamma$-far from $(2,1)$-free, then the algorithm FindCoordinate (Algorithm~\ref{Framework: testing $(1,2)$-freeness findcoordinate}) could find an index that belongs to the monotone structure with probability at least $1 - 1/(10\log n)$, using $\tilde O(\gamma^{-1}\log n)$ number of queries. This subroutine is used to determine the coordinate to query in the process of random binary search. If $k$ is not $\gamma$-deserted in the monotone sequence, the density of the monotone subsequence in any interval that contains this index is at least $\gamma$. This detailed procedure for adaptively finding the third element is shown in Algorithm~\ref{alg: Random Binary Search}.




\begin{algorithm}[h]
\caption{FindCoordinate$(f, [n], \epsilon, l)$}
\label{Framework: testing $(1,2)$-freeness findcoordinate}
\begin{algorithmic}[1]
\FOR{$t = 1, 2, \cdots, \lceil\log (n-l)\rceil$}
\STATE{Consider the intervals $L_t(l) = [l-r_{t}, l]$ and $R_t(l) = [l, l+r_t]$, where $r_t = 2^t$.}
\STATE{Sample from $R_t(l)$ uniformly at random for $\Theta(\log\log n/\epsilon)$ times. Sample from $L_t(l)$ uniformly at random for $\Theta(\log\log n/\epsilon)$ times.}
\IF{a $(1,2)$-tuple is found}
\STATE{return the second element.} 
\ENDIF
\ENDFOR\\
\STATE{return FAIL.} 
\end{algorithmic}
\end{algorithm}

\begin{algorithm}[h]
\caption{BinarySearch ($f, l, L_t(l), R_t(l)$)}
\label{alg: Random Binary Search}
\begin{algorithmic}[1]
\STATE{Initialization: $h\leftarrow 1, I_h\leftarrow R_t(l)$.}
\FOR{$m = 1,2, \cdots, \Theta(\log n)$}
\STATE{Call the subroutine FindCoordinate $(f, L_t(l)\cup I_h, \epsilon/\log n, l)$.}
\IF{this subroutine returns FAIL}
\STATE{obtain $x$ by randomly sampling a coordinate from $I_h$.}
\ELSE
\STATE{obtain $x$ from the subroutine FindCoordinate.}
\ENDIF
\IF{$f(x)\le f(i)$}
\STATE{update $I_{h+1}\leftarrow [x,b]$.}\label{line: update_l}
\ELSIF{$f(x)\ge f(j)$}
\STATE{update $I_{h+1}\leftarrow [a,x]$.}\label{line: update_r}
\ELSE
\STATE{return $x$.}
\ENDIF
\STATE{$h\leftarrow h+1$.}
\ENDFOR
\STATE{return FAIL.}
\end{algorithmic}
\end{algorithm}

\subsection{Analysis of tester for $(1,3,2)$-pattern}\label{sec: analysis_tester_1_3_2}

\subsubsection{$1$-gap tuples dominate}

When $1$-gap tuples dominate, algorithm~\ref{Framework: testing $(1,3,2)$-freeness when 1-gap terminates} outputs a $(1,3,2)$-tuple with high probability within $O(\epsilon^{-2}\log^3 n)$ number of queries.

\begin{lemma}\label{lem_gap_1_dominates}
For any $\epsilon>0$, given query access to a function $f:[n]\rightarrow \mathbb R$ which is $\epsilon$-far from $(1,3,2)$-free. For the case when $1$-gap tuples dominate, algorithm~\ref{Framework: testing $(1,3,2)$-freeness when 1-gap terminates} outputs a $(1,3,2)$ subsequence of $f$ with probability at least $19/20$. The query complexity of this algorithm is $O(\epsilon^{-2}\log^3 n)$.
\end{lemma}

\subsubsection{$2$-gap tuples dominate}

Now we will focus on the case when $2$-gap $(1,3,2)$-tuples dominate. Our analysis relies on the following property about $\gamma$-deserted element.

\begin{lemma}[\cite{newman2017testing}]\label{lem: gamma_deserted_lemma}
Let $S\subset I$ with $|S|\ge\epsilon |I|$. For every $\gamma<1$, at most $3\gamma(1-\epsilon)|I|/(1-\gamma)$ indices of $S$ are $\gamma$-deserted.
\end{lemma}

\begin{lemma}\label{lem:rtl_and_gamma_deserted}
Suppose the density of monotone sequence in $L_{t^*}(l)\cup R_{t^*}(l)$ is at least $\frac{\epsilon}{36\log n}$, for a fixed index $l$ and an integer ${t^*}$. Let $S$ be the monotone sequence in $R_{t^*}(l)$. Then the density of non-$\gamma$-deserted elements in $S$ is at least $1-1/(\zeta\log n)$, where $\gamma = \epsilon/(\eta\log n)$, where $\eta = 936$ and $\zeta = 52$ are constants.
\end{lemma}




Therefore, we could focus on the $(1,3,2)$-tuples satisfying that the third element is not $\gamma$-deserted in $R_t(l)$. The next lemma shows the correctness and query complexity of binary search as shown in Algorithm $\ref{alg: Random Binary Search}$.

\begin{lemma}\label{binary search lemma}
Suppose the density of monotone sequence in $L_{t^*}(l)\cup R_{t^*}(l)$ is at least $\frac{\epsilon}{36\log n}$, for a fixed index $l$ and an integer ${t^*}$. Given $(i,j)$, $i<j$ and $f(i)<f(j)$, $i, j\in L_{t^*}(l)$. With probability at least $25/26$, Algorithm \ref{alg: Random Binary Search} outputs an index $k$ satisfying that $(i,j,k)$ forms a $(1,3,2)$-pattern. The query complexity of this algorithm is $O(\epsilon^{-1}\log^3 n)$.
\end{lemma}


\begin{lemma}\label{lem_gap_2_dominate}
For any $\epsilon>0$, given query access to a function $f:[n]\rightarrow \mathbb R$ which is $\epsilon$-far from $(1,3,2)$-free. For the case when $2$-gap tuples dominate, Algorithm~\ref{Framework: testing $(1,3,2)$-freeness when 1-gap terminates} outputs a $(1,3,2)$ subsequence of $f$ with probability at least $19/20$. The query complexity of this algorithm is $O(\epsilon^{-2}\log^4 n)$.
\end{lemma}

With the above results, we are now ready to prove the conclusion as shown in Theorem~\ref{thm: test_132}. The missing proofs in this section are deferred to appendix.

\begin{theorem}\label{thm: test_132}[Restatement of Theorem~\ref{thm: test_132_main_result}]
For any $\epsilon>0$, there exists an adaptive algorithm that, given query access to a function $f:[n]\rightarrow \mathbb R$ which is $\epsilon$-far from $(1,3,2)$-free, outputs a $(1,3,2)$ subsequence of $f$ with probability at least $9/10$. The query complexity of this algorithm is $O(\epsilon^{-2}\log^4 n)$.
\end{theorem}

\section{Conclusion}
We firstly present a simple as well as efficient algorithm for testing monotone pattern with length $2$. We then provide an algorithm that achieves significant improvement in terms of the query complexity for testing $(1,3,2)$-tuple. It is worth investigating whether the newly proposed structure could be generalized to other forms? The interesting directions include designing a unified algorithm for any form of tuples, and proposing a simple adaptive algorithm for monotone-pattern testing that achieves optimal query complexity. 


\bibliographystyle{elsarticle-num}
\bibliography{references}
\newpage
\appendix

\section{Standard Binary Search Algorithm}
\begin{algorithm}[h]
\caption{Standard Binary Search}
\label{Framework: Standard Binary Search}
\begin{algorithmic}[1]
\STATE{Initialization: $h\leftarrow 1, I_h\leftarrow [a,b]$.}
\STATE{Let $x = a + |I_h|/3$.}
\IF{$f(x)\le f(i)$}
\STATE{update $I_{h+1}\leftarrow [x,b]$.}
\ELSIF{$f(x)\ge f(j)$}
\STATE{update $I_{h+1}\leftarrow [a,x]$.}
\ENDIF
\STATE{$h\leftarrow h+1$}
\end{algorithmic}
\end{algorithm}

\section{Missing proofs in section~\ref{sec: sec_alg_1_3_2_main}}

\subsection{Proof of Lemma~\ref{density_l_epsilon}}\label{proof_sec_6_1}

\begin{lemma}
If $f$ is $\epsilon$-far from $(1,2)$-free, then the ratio of $l\in [n]$ satisfying that $v_l\ge\epsilon/24$ is at least $\epsilon/24$.
\end{lemma}

\begin{proof}
From the definition of $A_{l,t}$, we have
\begin{align}
\sum_{l = 1}^{n}\sum_{t=0}^{\lfloor\log(n)\rfloor} |A_{l, t}|\ge\frac{1}{3}\sum_{t = 0}^{\lfloor\log(n)\rfloor} n_t\cdot 2^{t},
\end{align}
\noindent where $n_t$ is the number of tuples with $i_{2} - i_1 \in [2^{t}, 2^{t+1}]$.

\noindent Since $f$ is $\epsilon$-far from $(1,2)$-free, we have that
\begin{align}
    \Pr_{i\in [n]} [f(i)\neq g(i)]\ge\epsilon,
\end{align}
for any $(1,2)$-free function $g$.
From Lemma~\ref{gap_lemma_1}, we have that
\begin{align}
    \sum_{t = 0}^{\lfloor\log(n)\rfloor} n_t \ge \epsilon n/2.
\end{align}

\noindent Thus
\begin{align}
\sum_{l\in [n]}\sum_{t = 0} ^ {\lfloor\log(n)\rfloor} \delta_{l, t}\ge\sum_{t = 0}^{\lfloor\log(n)\rfloor} n_t/6 \ge\epsilon\cdot n/12.
\end{align}

\noindent Therefore, 
\begin{align}
\sum_{l\in [n]}v_l\ge\epsilon\cdot n/12.
\end{align}

\noindent Define $A = \{l\in [n]:v_l\ge \epsilon/24\}$, and $B = [n]\setminus A$. Then, we have that
\begin{align}
    |A| + \frac{\epsilon}{24} (n - |A|)\ge\sum_{l\in [n]}v_l\ge\epsilon n/12.
\end{align}

\noindent It implies that
\begin{align}
    |A|\ge\frac{\epsilon/24}{1-\epsilon/24}n\ge\epsilon n/24.
\end{align}
\end{proof}

\subsection{Proof of Theorem~\ref{lem_monotone_two}}\label{proof_sec_6_2}

\begin{theorem}
Suppose $l\in [n]$ is a an index with $\CumDensity$ at least $\epsilon/24$. For any $\epsilon>0$, given $l$ and query access to a function $f:[n]\rightarrow \mathbb R$ which is $\epsilon$-far from $(1,2)$-free, Algorithm~\ref{alg: testing $(1,2)$-freeness each epoch} outputs a length-$2$ monotone subsequence of $f$ with probability at least $9/10$. The query complexity of this algorithm is $O(\epsilon^{-1}\log n)$. 
\end{theorem}

\begin{proof}
If $f$ is $\epsilon$-far from $(1,2)$-free, then from Lemma \ref{density_l_epsilon} we know that, the ratio of $l\in [n]$ satisfying the following equation is at least $\epsilon/24$,
\begin{align}\label{eqn: density_sum}
\delta_{l,1} + \delta_{l,2} + \cdots + \delta_{l,h}\ge\epsilon/24,
\end{align}
where $h = \lceil\log(n-l)\rceil$. Therefore, by repeating this procedure $O(1/\epsilon)$ number of times, the algorithm could find an index $l$ that satisfies Eq.~(\ref{eqn: density_sum}) with high constant probability.

Recall that $\delta_{l, t} = |A_{l, t}|/2^{t+1}$, $A_{l, t} = \{(i_1, i_2)\in U: \text{width}(i_1, i_{2}) = t, l\in [i_1-(i_{2}-i_1)/3, i_{2}+(i_{2}-i_{1})/3]\}$, and $U$ is the set of disjoint $(i_1, i_2, \cdots, i_k)$-sequences of $f$. Define $B_t = L_t(l)\cup R_t(l)$. We analyze according to the following two cases.

\noindent Case 1:
If there does not exist a block such that the density of $(1,2)$ in this interval is at least $\epsilon/50$.

\noindent If the density of $(1,2)$-tuple in all the blocks $B_1, B_2, \cdots, B_{h}$ is at most $\epsilon$, then the structure is dominated by growing suffix from the definition of $(\alpha, \beta)$-growing suffix. Let $D_t(l) = \{i_2:(i_1,i_2)\in A_{l,t}\}$ if $ t = 3, 6, 9, 12, \cdots$, and $D_t(l) =\emptyset$ otherwise. Then $D_t(l)\subset R_{t+1}(l)$, and satisfies the property that $f(b)<f(b')$ if $b\in D_t(l)$ and $b'\in D_{t'}(l)$, where $t, t'\in [h]$ and $t<t'$. It is satisfied that
\begin{align}
    |D_t(l)|/|R_{t+1}(l)|\le\alpha,
\end{align}
and that
\begin{align}
    \sum_{t=1}^{h} |D_t(l)|/|R_{t+1}(l)| = \sum_{t=1}^{h}\delta_t(l) \ge\beta,
\end{align}
\noindent where $\alpha = \epsilon/50$, and $\beta = \epsilon/24$.

\noindent Let $\mathcal{E}_t$ be the event that at least one index that belongs to $D_t(l)$ is obtained by sampling from $R_t(l)$. Therefore, we have that
\begin{align}
   \mathbb E[\sum_{t=1}^{h}\mathbbm{1}\{\mathcal E_{t}\}] \ge 
    \sum_{t=1}^{h}T\delta_t(l)\ge C/2,
\end{align}

\noindent where the first inequality is due to $T\delta_t(l)\le 1$ for all $t\in [h]$.

\noindent Therefore, $\sum_{t=1}^{h}\mathbbm{1}\{\mathcal E_{t}\}\ge 2$ with high probability. It implies that at least two indexes could be found from $D_{j_1}(l)$ and $D_{j_2}(l)$ separately. Therefore, $(1,2)$-pattern could be found from the set of blocks that lie at the right of $l$ with high probability.

\noindent Case 2: If there exists a block $B_t$ such that the density of $(1,2)$ in this block is at least $\Theta(\epsilon)$. From the definition, it implies that the density of $1$ in $L_t(l)$ is at least $\Theta(\epsilon)$, and the density of $2$ in $R_t(l)$ is at least $\Theta(\epsilon)$. Therefore, by sampling $O(1/\epsilon)$ times from each block, we could find a $(1,2)$-pattern by concatenating the left block and the right block, with high probability.

\noindent The query complexity of our algorithm is $O(\epsilon^{-1}\log n)$ from the design scheme.
\end{proof}

\section{Missing proofs in Section~\ref{sec: monotone_1_3_2}}

\subsection{Proof of Lemma~\ref{gap_lemma_1}}



\begin{lemma}[Proposition 2.2 of \cite{newman2017testing}]
If $f$ is $\epsilon$-far from $\pi$-free, where $\pi$ is a permutation of $[m]$. Then there exists at least $\epsilon n/m$ disjoint permutation in the form of $\pi$.
\end{lemma}

\begin{proof}
Let $T$ be a maximal disjoint set of $\pi$-tuples in $f$. Let $E(T) = \bigcup_{(i_1, i_2, \cdots, i_m)\in T}\{i_1, i_2, \cdots, i_m\}$. Then, $f$ restricted to the set $[n]\setminus E(T)$ is $\pi$-free. $f$ could be modified as a $\pi$-free function over $[n]$, if $f(i)$ for each $i\in I$ is replaced by the function value of the closest element that does not belong to $I$. That is, $f(i)$ for each $i\in E(T)$ is replaced by $f(j)$, where $j$ is the largest integer prior to $i$ satisfying that $j\notin E(T)$. If such a $j$ does not exist, then $f(i)$ is replaced by $\max_{a\in [n]} f(a)$.
Since $f$ is $\epsilon$-far from being $\pi$-free, we have $|E(T)|\ge\epsilon n$. Therefore, $|T|\ge\epsilon n/m$. 
\end{proof}

\subsection{Proof of Lemma~\ref{lem:ijk_monotone}}


\begin{lemma}
Let $T$ be the output of the subroutine GreedyDisjointTuples$+$.
For a given $l$, consider the set of tuples $(i, j, k)\in T$ satisfying that $j\le l \le k$. We denote $A = \{(i, j, k)\in T: j\le l \le k\}$, then the elements $k\in A$ form an increasing sequence.
\end{lemma}
\begin{proof}
We want to show the following claim. Consider any two tuples $(i_1,j_1,k_1)\in A$, and $(i_2,j_2,k_2)\in A$. If $k_1<k_2$, then we have that $f(k_1)\le f(k_2)$.

Since $(i_2,j_2,k_2)\in A$, we know that $(i_2,j_2,k_2)$ belongs to disjoint $(1,3,2)$-pattern, satisfying that $f(i_2)<f(k_2)<f(j_2)$, and that $i_2<j_2<l<k_2$. Suppose that $l<k_1<k_2$, if $f(k_1)>f(k_2)$, then according to the subroutine GreedyDisjointTuples$+$, the algorithm will chose $(i_2, k_1, k_2)$ to form a $(1,3,2)$-pattern. This leads to a contradiction since the algorithm indeed selects $(i_2,j_2,k_2)$ instead of $(i_2, k_1, k_2)$ to form a $(1,3,2)$-pattern.

This implies that for $k_1$ and $k_2$ satisfying that $l<k_1<k_2$, we have that $f(k_1)\le f(k_2)$. Therefore, the set of $k$ form an increasing sequence.


\end{proof}

\subsection{Proof of Lemma~\ref{gap_lemma_2}}


\begin{lemma}
Let $f: [n]\rightarrow \mathbb R$, and let $T_0\subset [n]^3$ be a set of disjoint $(1,3,2)$-tuples of $f$ with length $3$. Then there exist $c\in [2]$ and a set $T\subset [n]^3$ of disjoint $c$-gap $(1,3,2)$-tuples with $E(T)\subset E(T_0)$, satisfying that $|T|\ge|T_0|/6$, where $E(T) = \cup_{(i,j,k)\in T} \{i,j,k\}$ is the set of indices of subsequences in $T$.
\end{lemma}


\begin{proof}
Let $T'$ reprsent the set of disjoint $(1,3,2)$-tuples. 
Suppose $|T'|<|T_0|/3$. Then, there exists a tuple $(i_1, i_2, i_3)\in T_0$ satisfying that $\{i_1, i_2, i_3\}\cap E(T') = \emptyset$. Recall that GreedyDisjointTuples$+$ increases the size of $T'$ throughout the execution, $\{i_1, i_2, i_{3}\}\cap T' = \emptyset$. When the algorithm finds $k$, a $(1,3,2)$-tuple disjoint from $T'$ could be found, then $k$ is added to $T'$. This leads to a contradiction. Therefore, $|T'|\ge|T_0|/3$, and $|T|\ge|T_0|/6$.
\end{proof}

\section{Missing proofs in Section~\ref{sec: analysis_tester_1_3_2}}

\subsection{Proof of Lemma~\ref{lem_gap_1_dominates}}


\begin{lemma}
For any $\epsilon>0$, given query access to a function $f:[n]\rightarrow \mathbb R$ which is $\epsilon$-far from $(1,3,2)$-free. For the case when $1$-gap tuples dominate, algorithm~\ref{Framework: testing $(1,3,2)$-freeness when 1-gap terminates} outputs a $(1,3,2)$ subsequence of $f$ with probability at least $19/20$. The query complexity of this algorithm is $O(\epsilon^{-2}\log^3 n)$.
\end{lemma}

\begin{proof}
Suppose that the algorithm finds an integer $t$ satisfying that the density of $i$ in $L_t(l)$ is at least $\Theta(\epsilon/\log n)$, and the density of $(j,k)$ in $R_t(l)$ is at least $\Theta(\epsilon/\log n)$. We then divide the blocks into top, middle and bottom blocks based on the function value. The algorithm could then find $i$ such that $i$ lies in the middle block with probability at least $\epsilon/\log n$. Besides, the density of $(j,k)$ that could concatenate with the selected $i$ to form a $(1,3,2)$-tuple is lower bounded by $\Omega(\epsilon/\log n)$. The algorithm could find the required $(j,k)$ with high constant probability using $O(\log^2 n/\epsilon)$ queries. The query complexity of this algorithm is $O(\epsilon^{-2}\log^3 n)$.
\end{proof}

\subsection{Proof of Lemma~\ref{lem:rtl_and_gamma_deserted}}

\begin{lemma}
Suppose the density of monotone sequence in $L_{t^*}(l)\cup R_{t^*}(l)$ is at least $\frac{\epsilon}{36\log n}$, for a fixed index $l$ and an integer ${t^*}$. Let $S$ be the monotone sequence in $R_{t^*}(l)$. Then the density of non-$\gamma$-deserted elements in $S$ is at least $1-1/(\zeta\log n)$, where $\gamma = \epsilon/(\eta\log n)$, where $\eta = 936$ and $\zeta = 52$ are constants.
\end{lemma}
\begin{proof}
Define $S' = \{i\in S: \text{i is $\gamma$-deserted in S}\}$. For a fixed index $l$, consider $S = \{k: k\in A\}$, where $A = \{(i, j, k)\in T: j\le l \le k\}$. That is, $S$ is the monotone sequence in $R_{t^*}(l)$. From Lemma~\ref{lem: gamma_deserted_lemma}, we have that $|S'|\le 3\gamma(1-\frac{\epsilon}{36\log n})|L_{t^*}(l)\cup R_{t^*}(l)|/(1-\gamma)\le |S|/(52\log n)$. We thus have that $|S\setminus S'| = |S| - |S'| \ge (1-1/(52\log n))|S|$. Then the density of non-$\gamma$-deserted elements in the monotone sequence is at least $(1-1/(52\log n))$.
\end{proof}

\subsection{Proof of Lemma~\ref{binary search lemma}}


\begin{lemma}
Suppose the density of monotone sequence in $L_{t^*}(l)\cup R_{t^*}(l)$ is at least $\frac{\epsilon}{36\log n}$, for a fixed index $l$ and an integer ${t^*}$. Given $(i,j)$, $i<j$ and $f(i)<f(j)$, $i, j\in L_{t^*}(l)$. With probability at least $25/26$, Algorithm \ref{alg: Random Binary Search} outputs an index $k$ satisfying that $(i,j,k)$ forms a $(1,3,2)$-pattern. The query complexity of this algorithm is $O(\epsilon^{-1}\log^3 n)$.
\end{lemma}

\begin{proof}
The key point is the assumption that 
the element $k$ that we are searching for is not a $\gamma$- deserted element in $R_{t^*}(l)$, where $\gamma = \frac{\epsilon}{936 \log n}$. Let $I_h$ be the search range at iteration $h$, and let $A$ be the set of monotone sequence in $R_{t^*}(l)$. Initially, $I_1 = R_{t^*}(l)$. 

The density of $k$ that cross $l$ in $L_{t^*}(l)\cup R_{t^*}(l)$ is at least $\Theta(\epsilon/\log n)$. Lemma~\ref{lem:ijk_monotone} implies that the density of monotone sequence in $L_{t^*}(l)\cup R_{t^*}(l)$  is at least $\Theta(\epsilon/\log n)$. That is, $|A|\ge \Theta(\epsilon/\log n) |L_{t^*}(l)\cup R_{t^*}(l)|$. Since $k$ is not a $\gamma$-deserted element in $L_{t^*}(l)\cup R_{t^*}(l)$. From the definition of deserted element, we know that the density of monotone sequence in any interval that contains $k$ is at least $\gamma$. That is, $|A\cap I|\ge\gamma |I|$ holds for any $I\subset L_{t^*}(l)\cup R_{t^*}(l)$. 

Let $\mathcal E_h$ be the event that an element that belongs to $A$ is found at iteration $h$. Lemma~\ref{lem:rtl_and_gamma_deserted} implies that the density of non-$\gamma$-deserted elements in the monotone sequence $A$ is at least $(1-1/(52\log n))$.

The approach is as follows: call the subroutine FindCoordinate$(f, L_{t^*}(l)\cup I_h, \gamma, l)$. The non-$\gamma$-deserted element $k\in L_{t^*}(l)\cup I_h\subset L_{t^*}(l)\cup R_{t^*}(l)$, it implies that $f|_{(L_{t^*}(l)\cup I_h)}$ is $(\gamma/5)$-far from $(2,1)$-free. The algorithm FindCoordinate could find an element that belongs to $A$ from $L_{t^*}(l)\cup I_h$ with probability at least $1 - 1/(52\log n)$. Then,
\begin{align}
    \Pr(\bar {\mathcal E}_h) \le \Pr(\bar {\mathcal E}_{h-1}) + \Pr({\mathcal E}_{h-1})\Pr(\bar {\mathcal E}_{h}|{\mathcal E}_{h-1})
    \le 1/(26\log n).
\end{align}

Therefore, the probability that $\mathcal E_h$ occurs is at least $1-1/(26\log n)$.
\noindent By a union bound over all iterations, we have that
\begin{align}
    \Pr(\text{binary search succeeds}) &= 
    \Pr(\cap_h {\mathcal E}_h) \nonumber\\
    &= 1 - \Pr(\cup_h \bar {\mathcal E}_h)
    \ge 1 - \log(n) \cdot 1/(26\log n)\ge 25/26.
\end{align}

\noindent The query complexity of FindCoordinate is $O(\gamma^{-1}\log n) = O(\log^2 n/\epsilon)$. Therefore, the total query complexity of this procedure is upper bounded by $O(\epsilon^{-1}\log^3 n)$.

\end{proof}

\subsection{Proof of Lemma~\ref{lem_gap_2_dominate}}


\begin{lemma}
For any $\epsilon>0$, given query access to a function $f:[n]\rightarrow \mathbb R$ which is $\epsilon$-far from $(1,3,2)$-free. For the case when $2$-gap tuples dominate, Algorithm~\ref{Framework: testing $(1,3,2)$-freeness when 1-gap terminates} outputs a $(1,3,2)$ subsequence of $f$ with probability at least $19/20$. The query complexity of this algorithm is $O(\epsilon^{-2}\log^4 n)$.
\end{lemma}

\begin{proof}
With Lemma \ref{density_l_epsilon} and Lemma \ref{binary search lemma}, we could analyze the query complexity of the tester. From Lemma \ref{density_l_epsilon}, the ratio of $l\in [n]$ satisfying that $v_l\ge\epsilon/24$ is at least $\epsilon/24$. Suppose $l^*$ is the index with $v_{l^*}\ge\epsilon/24$. Then, there exists an integer $t^*$, satisfying that the density of tuples crossing ${l^*}$ in interval $[{l^*}-2t^*, {l^*}+t^*]$ is at least $\epsilon_1 = \epsilon/(36\log n)$. The query complexity of finding $(i,j)$ at the left of $l^*$ using the TestMonotoneEpoch algorithm is $O(\epsilon_1^{-1}\log n)$. From Lemma \ref{binary search lemma}, the query complexity of binary search for finding $k$ is $O(\epsilon^{-1}\log^3 n)$. Therefore, the total query complexity is
$\epsilon^{-1}\log n\left(O(\epsilon_1^{-1}\log n) + O(\epsilon^{-1}\log^3 n)\right) =  O(\epsilon^{-2}\log^{4}n)$.

\end{proof}

\subsection{Proof of Theorem~\ref{thm: test_132}}


\begin{theorem}[Restatement of Theorem~\ref{thm: test_132_main_result}]
For any $\epsilon>0$, there exists an adaptive algorithm that, given query access to a function $f:[n]\rightarrow \mathbb R$ which is $\epsilon$-far from $(1,3,2)$-free, outputs a $(1,3,2)$ subsequence of $f$ with probability at least $9/10$. The query complexity of this algorithm is $O(\epsilon^{-2}\log^4 n)$.
\end{theorem}

\begin{proof}
Combining  Lemma~\ref{lem_gap_1_dominates} and Lemma~\ref{lem_gap_2_dominate}, with probability at least $1-1/20-1/20=9/10$, our algorithm outputs a $(1,3,2)$-tuple. The query complexity of our algorithm is $O(\epsilon^{-2}\log^3 n) + O(\epsilon^{-2}\log^{4}n) = O(\epsilon^{-2}\log^{4}n)$.

\end{proof}


\end{document}

%% file: commands.tex
\usepackage[utf8]{inputenc}
\usepackage{microtype}
\usepackage[T1]{fontenc}    
\usepackage{url}            
\usepackage{booktabs}       
\usepackage{amsfonts}       
\usepackage{nicefrac}       
\usepackage{amsthm}

\usepackage{tikz}
\theoremstyle{plain}

\newtheorem{lemma}{Lemma}
\newtheorem{theorem}{Theorem}
\newtheorem{definition}{Definition}

\newcommand{\CumDensity}{\textit{cumulative density}}

\usepackage{afterpage}
\usepackage{prettyref}
\usepackage{framed}
\newrefformat{eq}{\savehyperref{#1}{Eq. \textup{(\ref*{#1})}}}
\newrefformat{eqn}{\savehyperref{#1}{Eq.~\textup{(\ref*{#1})}}}
\newrefformat{lem}{\savehyperref{#1}{Lemma~\ref*{#1}}}
\newrefformat{defi}{\savehyperref{#1}{Definition~\ref*{#1}}}
\newrefformat{lemma}{\savehyperref{#1}{Lemma~\ref*{#1}}}
\newrefformat{def}{\savehyperref{#1}{Definition~\ref*{#1}}}
\newrefformat{line}{\savehyperref{#1}{Line~\ref*{#1}}}
\newrefformat{thm}{\savehyperref{#1}{Theorem~\ref*{#1}}}
\newrefformat{corr}{\savehyperref{#1}{Corollary~\ref*{#1}}}
\newrefformat{cor}{\savehyperref{#1}{Corollary~\ref*{#1}}}
\newrefformat{sec}{\savehyperref{#1}{Section~\ref*{#1}}}
\newrefformat{app}{\savehyperref{#1}{Appendix~\ref*{#1}}}
\newrefformat{assum}{\savehyperref{#1}{Assumption~\ref*{#1}}}
\newrefformat{ex}{\savehyperref{#1}{Example~\ref*{#1}}}
\newrefformat{fig}{\savehyperref{#1}{Figure~\ref*{#1}}}
\newrefformat{alg}{\savehyperref{#1}{Algorithm~\ref*{#1}}}
\newrefformat{rem}{\savehyperref{#1}{Remark~\ref*{#1}}}
\newrefformat{conj}{\savehyperref{#1}{Conjecture~\ref*{#1}}}
\newrefformat{prop}{\savehyperref{#1}{Proposition~\ref*{#1}}}
\newrefformat{proto}{\savehyperref{#1}{Protocol~\ref*{#1}}}
\newrefformat{prob}{\savehyperref{#1}{Problem~\ref*{#1}}}
\newrefformat{claim}{\savehyperref{#1}{Claim~\ref*{#1}}}
\newrefformat{que}{\savehyperref{#1}{Question~\ref*{#1}}}
\newrefformat{op}{\savehyperref{#1}{Open Problem~\ref*{#1}}}
\newrefformat{fn}{\savehyperref{#1}{Footnote~\ref*{#1}}}

\usepackage{xspace}
\usepackage{amsmath}
\usepackage{amsthm}
\usepackage{bbm}
\usepackage{algorithm}
\usepackage[algo2e, vlined, ruled]{algorithm2e}
\usepackage{algorithmic}
\usepackage{mathrsfs}
\usepackage{bm}
\usepackage{makecell}
\usepackage{tabulary}